\documentclass[12pt]{article}
\usepackage{amsmath,amsthm,amsfonts}
\usepackage{fullpage}

\newtheorem{theorem}{Theorem}

\newtheorem{lemma}[theorem]{Lemma}

\newcounter{def_counter}
\theoremstyle{definition}
\newtheorem{defn}[def_counter]{Definition}

\newcommand{\N}{\mathbb N}

\newcommand{\norm}[1]{\|#1\|}

\title{Finding the growth rate of a regular language in polynomial time}

\author{Dalia Krieger, Narad Rampersad, and Jeffrey Shallit\\
School of Computer Science\\
University of Waterloo\\
Waterloo, Ontario  N2L 3G1\\
Canada\\
{\tt d2kriege@cs.uwaterloo.ca} \\
{\tt nrampersad@math.uwaterloo.ca} \\
{\tt shallit@graceland.uwaterloo.ca}}

\begin{document}
\date{\today}
\maketitle

\begin{abstract}
We give an $O(n^3+n^2 t)$ time algorithm to determine whether
an NFA with $n$ states and $t$ transitions accepts a language
of polynomial or exponential growth.  We also show that
given a DFA accepting a language of polynomial growth,
we can determine the order of polynomial growth in
quadratic time.
\end{abstract}

\section{Introduction}

Let $L \subseteq \Sigma^*$ be a language. If there exists a polynomial $p(x)$ such that $|L \cap
\Sigma^m| \leq p(m)$ for all $m \geq 0$, then $L$ has \emph{polynomial growth}. Languages of
polynomial growth are also called \emph{sparse} or \emph{poly-slender}.

If there exists a real number $r > 1$ such that $|L \cap \Sigma^m| \geq r^m$ for infinitely many $m
\geq 0$, then $L$ has \emph{exponential growth}.  Languages of exponential growth are also called
\emph{dense}.

If there exist words
$w_1,w_2,\ldots,w_k \in \Sigma^*$ such that $L \subseteq
w_1^*w_2^*\cdots w_k^*$, then $L$ is called a \emph{bounded language}.

Ginsburg and Spanier \cite{GS64} (see Ginsburg \cite[Chapter~5]{Gin66})
proved many deep results concerning the structure of bounded context-free
languages.  One significant result \cite[Theorem~5.5.2]{Gin66} is that
determining if a context-free grammar generates a bounded language is
decidable.  However, although it is a relatively straightforward consequence
of their work, they did not make the following connection between the
bounded context-free languages and those of polynomial growth.

\begin{theorem}
\label{bounded}
A context-free language is bounded if and only if it has polynomial growth.
\end{theorem}

Curiously, this result has been independently discovered at least six times:
namely, by Trofimov \cite{Tro81}, Latteux and Thierrin \cite{LT84},
Ibarra and Ravikumar \cite{IR86}, Raz \cite{Raz97},
Incitti \cite{Inc01}, and Bridson and Gilman \cite{BG02}.
A consequence of all of these proofs is that a context-free language has
either polynomial or exponential growth; no intermediate growth is possible.

The particular case of the bounded regular languages was also studied
by Ginsburg and Spanier \cite{GS66}, and subsequently by Szilard, Yu,
Zhang, and Shallit \cite{SYZS92} (see also \cite{IRS00}).
It follows from the more general decidability result of Ginsburg and Spanier
that there is an algorithm to determine whether a regular language has
polynomial or exponential growth (see also \cite[Theorem~5]{SYZS92}).
Ibarra and Ravikumar \cite{IR86} observed that the algorithm
of Ginsburg and Spanier runs in polynomial time for NFA's, but
they gave no detailed analysis of the runtime.
Here we specialize the algorithm of Ginsburg and Spanier
to the case of regular languages, and we give particular attention to the
runtime of this algorithm.  We also show how, given a DFA accepting
a language of polynomial growth as input, one may determine the precise
order of polynomial growth in polynomial time.

\section{Polynomial vs. exponential growth}

In this section we give an $O(n^3+n^2 t)$ time algorithm to determine whether
an NFA with $n$ states and $t$ transitions accepts a language of polynomial
or exponential growth.

\begin{theorem}
\label{mod_gins} Given a NFA $M$, it is possible to test whether $L(M)$
is of polynomial or exponential growth in $O(n^3+n^2 t)$ time, where
$n$ and $t$ are the number of states and transitions of $M$ respectively.
\end{theorem}

Let $M = (Q,\Sigma,\delta,q_0,F)$ be an NFA.  We assume that every state of
$M$ is both accessible and co-accessible, i.e., every state of $M$ can
be reached from $q_0$ and can reach a final state.  For each state $q \in Q$,
we define a new NFA $M_q = (Q,\Sigma,\delta,q,\{q\})$ and write
$L_q = L(M_q)$.

Following Ginsburg and Spanier, we say that a language $L \subseteq \Sigma^*$
is \emph{commutative} if there exists $u \in \Sigma^*$ such that
$L \subseteq u^*$.

The following two lemmas have been obtained in more generality in all of
the previously mentioned proofs of Theorem~\ref{bounded} (compare also
Lemmas~5.5.5 and 5.5.6 of Ginsburg \cite{Gin66}, or in the case
of regular languages specified by DFA's, Lemmas~2 and 3 of
Szilard et al.\ \cite{SYZS92}).

\begin{lemma}
\label{comm1}
If $L(M)$ has polynomial growth, then for every $q \in Q$,
$L_q$ is commutative.
\end{lemma}

\begin{proof}
A classical result of Lyndon and Sch\"utzenberger \cite{LS62} implies that if a set of words $X$
does not satisfy $X \subseteq u^*$ for any word $u$, then there exist $x,y \in X$ such that $xy
\neq yx$. Suppose then that $L(M)$ has polynomial growth, but for some $L_q$ there exists $x,y \in
L_q$, $xy \neq yx$. Let $v$ be any word such that $q \in \delta(q_0,v)$, and let $v'$ be any word
such that $\delta(q,v') \cap F \neq \emptyset$. Then for every $m \geq 0$, the set $v (xy+yx)^m v'$
consists of $2^m$ distinct words of length $|vv'| + m|xy|$ in $L(M)$. It follows that $L(M)$ has
exponential growth, contrary to our assumption.
\end{proof}

\begin{lemma}
\label{comm2}
If for every $q \in Q$, $L_q$ is commutative,
then $L(M)$ has polynomial growth.
\end{lemma}

\begin{proof}
We prove by induction on the number $n$ of states of $M$ that the
hypothesis of the lemma implies that $L(M)$ is bounded.  It is well-known
that any bounded language has polynomial growth (see, for example,
\cite[Proposition~1]{Raz97}).  Clearly the result holds for $n = 1$.
We suppose then that $n > 1$.

Let $Q' = Q \setminus \{q_0\}$, $F' = F \setminus \{q_0\}$,
and $\delta'(q,a) = \delta(q,a) \setminus \{q_0\}$ for all $q \in Q'$ and
$a \in \Sigma$.  For each $q \in Q'$, we define an NFA
$N_q = (Q',\Sigma,\delta',q,F')$, and we write $A_q = L(N_q)$.
Applying the induction hypothesis to $N_q$, we conclude that $A_q$
is bounded.

The key observation is that $L(M) = L_1 \cup L_2$, where
\[
L_1 = \bigcup_{a \in \Sigma} \left( \bigcup_{q \in \delta(q_0,a)}
L_{q_0} a A_q \right),
\]
and
\[
L_2 = \begin{cases} L_{q_0}, & \text{if } q_0 \in F; \\
                    \emptyset, & \text{if } q_0 \notin F.
      \end{cases}
\]
By assumption, $L_{q_0} \subseteq u^*$ for some $u \in \Sigma^*$,
and, as previously noted, by the induction hypothesis
each of the languages $A_q$ is bounded.
It follows that $L(M)$ is a finite union of bounded languages,
and hence is itself bounded.  We conclude that $L(M)$ has polynomial growth,
as required.
\end{proof}

We now are ready to prove Theorem~\ref{mod_gins}.

\begin{proof}
Let $n$ denote the number of states of $M$.  The idea is as follows.
For every $q \in Q$, if $L_q$ is commutative, then there exists
$u \in \Sigma^*$ such that $L_q \subseteq u^*$.  For any $w \in L_q$,
we thus have $w \in u^*$.  If $z$ is the primitive root of $w$, then
$z$ is also the primitive root of $u$.  If $L_q \subseteq z^*$, then
$L_q$ is commutative.  On the other hand, if $L_q \not\subseteq z^*$,
then $L_q$ contains two words with different primitive roots,
and is thus not commutative.  This argument leads to the following
algorithm.

\begin{quotation}
\noindent For each $q \in Q$ we perform the following steps.
  \begin{itemize}
  \item Construct the NFA $M_q$ accepting $L_q$.  This takes $O(n+t)$ time.
  \item Find a word $w \in L(M_q)$, where $|w| < n$.  If $L(M_q)$ is non-empty,
  such a $w$ exists and can be found in $O(n+t)$ time.
  \item Find the primitive root of $w$, i.e., the shortest word $z$ such
  that $w = z^k$ for some $k\geq 1$.  This can be done in $O(n)$ time
  using the Knuth--Morris--Pratt algorithm.  To find the primitive root of
  $w = w_1 \cdots w_\ell$, use Knuth--Morris--Pratt to find the first
  occurrence of $w$ in $w_2 \cdots w_\ell w_1 \cdots w_{\ell-1}$.  If the first
  occurrence begins at position $i$, then $z = w_1 \cdots w_{i-1}$ is the
  primitive root of $w$.
  \item Apply the cross product construction to obtain an
  NFA $M'$ that accepts $L_q \setminus z^*$.  The NFA $M'$
  has $O(n^2)$ states and $O(nt)$ transitions.
  \item Test whether $L(M')$ is empty or not.  If $L(M')$ is non-empty,
  then by Lemma~\ref{comm1} the growth of $L(M)$ is exponential.
  If $L(M')$ is empty, then $L_q$ is commutative.  This step takes
  $O(n^2 + nt)$ time.
  \end{itemize}
If for all $q \in Q$ we have verified that $L_q$ is commutative, then by
Lemma~\ref{comm2} $L(M)$ has polynomial growth.
\end{quotation}
The runtime of this algorithm is $O(n^3+n^2t)$.
\end{proof}

\section{Finding the exact order of polynomial growth}

In this section we show that given a DFA accepting a language of
polynomial growth, it is possible to efficiently
determine the exact order of polynomial growth.  We give two
different algorithms: one combinatorial, the other algebraic.

Szilard et al.\ \cite[Theorem~4]{SYZS92} proved a weaker result: namely, that given a regular
language $L$ and an integer $d \geq 0$ it is decidable whether $L$ has $O(m^d)$ growth.  However,
even if $L$ is specified by a DFA, their algorithm takes exponential time.

\subsection{A combinatorial algorithm}

\begin{theorem}
\label{order} Given a DFA $M$ with $n$ states such that $L(M)$ is of polynomial growth, it is
possible to determine the exact order of polynomial growth in $O(n^2)$ time.
\end{theorem}

\begin{proof}
Let $M = (Q,\Sigma,\delta,q_0,F)$.  Again we assume that every state
of $M$ is both accessible and co-accessible.  Since $L(M)$ is of polynomial growth,
by Lemma~\ref{comm1} $M$ has the property that for every $q \in Q$
there exists $u \in \Sigma^*$ such that $L_{q} \subseteq u^*$.

Since $M$ is deterministic, for any state $q$ of $M$,
if there exists a non-empty word $w$ that takes $M$ from state $q$ back
to $q$, the smallest such word $w$ is unique.  There is also a unique
cycle of states of $M$ associated with such a word $w$, and all such cycles
in $M$ are disjoint.

We now \emph{contract} (in the standard graph-theoretical sense)
each such cycle to a single vertex and mark this vertex as \emph{special}.
If any vertex on a contracted cycle was final, we also mark the new special
vertex as final.  Since all the cycles in $M$ are disjoint, after contracting
all of them, the transition graph of the
automaton $M$ now becomes a \emph{directed acyclic graph} (DAG) $D$.
A path in $D$ from the start vertex to a final vertex that visits
special vertices $Q_1,Q_2,\ldots,Q_k$ corresponds to a family
of words in $L(M)$ of the form
\begin{equation}
\label{decomp}
x_1 y_1^* x_2 y_2^* \cdots x_k y_k^* x_{k+1},
\end{equation}
where the $y_i$'s are words labeling the cycles in $M$ corresponding
to the $Q_i$'s in $D$.  Note that if a cycle in $M$ is of size $t$, there
could be up to $t$ possible choices for the corresponding $y_i$.

There are only finitely many paths in $D$, and only finitely many
choices for the $x_i$'s and $y_i$'s in a decomposition of the form
given by (\ref{decomp}).  It follows that $L(M)$ is a finite union
of languages of the form $x_1 y_1^* x_2 y_2^* \cdots x_k y_k^* x_{k+1}$.
We have thus recovered the characterization of Szilard et
al.\ \cite{SYZS92}.  It is well-known that any language of this form
has $O(m^{k-1})$ growth (see, for example, \cite[Lemma~4]{SYZS92}).

Consider a path through $D$ from the start vertex to a final vertex that visits the maximum number
$d$ of special vertices. Then we may conclude that the order of growth of $L(M)$ is
$\Theta(m^{d-1})$.  This observation leads to our desired algorithm.

We first identify all the cycles in $M$ and contract them to obtain
a DAG $D$, as previously described.  It remains to find
a path through $D$ from the start vertex to a final vertex that visits
the largest number of special vertices.  The LONGEST PATH problem
for general graphs is NP-hard; however, in the case of a DAG,
it can be solved in linear time by a simple dynamic programming algorithm.
To obtain our result, we modify this dynamic programming algorithm by
adjusting our distance metric so that the length of a path is not
the number of edges on it, but rather the number of special vertices on the
path.  The most computationally intensive part of this algorithm is
finding and contracting the cycles in $M$, which can be done in $O(n^2)$ time.
\end{proof}

\subsection{An algebraic approach}

We now consider an algebraic approach to determining whether the order of growth is polynomial or
exponential, and in the polynomial case, the order of polynomial growth. Let $M =
(Q,\Sigma,\delta,q_0,F)$, where $|Q| = n$, and let $A = A(M) = (a_{ij})_{1\leq i,j\leq n}$ be the
\emph{adjacency matrix} of $M$, that is, $a_{ij}$ denotes the number of paths of length 1 from
$q_i$ to $q_j$. Then $(A^m)_{i,j}$ counts the number of paths of length $m$ from $q_i$ to $q_j$.
Since a final state is reachable from every state $q_j$, the order of growth of $L(M)$ is the order
of growth of $A^m$ as $m\rightarrow\infty$. This order of growth can be estimated using nonnegative
matrix theory.

\begin{theorem} [Perron-Frobenius]\label{thm:P_F}
Let $A$ be a nonnegative square matrix, and let $r$ be the spectral radius of $A$, i.e., $r =
\max\{|\lambda|:\lambda\textrm{ is an eigenvalue of }A\}$. Then
\begin{enumerate}
    \item $r$ is an eigenvalue of $A$;
    \item there exists a positive integer $h$ such that any eigenvalue $\lambda$ of $A$
    with $|\lambda| = r$ satisfies $\lambda^h = r^h$.
\end{enumerate}
\end{theorem}
For more details, see \cite[Chapters 1, 3]{Minc88}.

\begin{defn}\label{def:P-F}\rm
The number $r = r(A)$ described in the above theorem is called the \emph{Perron-Frobenius
eigenvalue of $A$}. The \emph{dominating Jordan block} of $A$ is the largest block in the Jordan
decomposition of $A$ associated with $r(A)$.
\end{defn}

\begin{lemma}\label{lem:r>=1}
Let $A$ be a nonnegative $n\times n$ matrix over the integers. Then either $r(A) = 0$ or $r(A) \geq
1$.
\end{lemma}
\begin{proof}
Let $r(A) = r,\lambda_1\cdots,\lambda_\ell$ be the distinct eigenvalues of $A$, and suppose that $r
< 1$. Then $\lim_{m\rightarrow\infty}r^m = \lim_{m\rightarrow\infty}\lambda_i^m = 0$ for all $i =
1,\ldots,\ell$, and so $\lim_{m\rightarrow\infty}A^m = 0$ (the zero matrix). But $A^m$ is an
integral matrix for all $m\in\N$, and the above limit can hold if and only if $A$ is nilpotent,
i.e., $r = \lambda_i = 0$ for all $i = 1,\ldots,\ell$.
\end{proof}

\begin{lemma}\label{lem:growthorder}
Let $A$ be a nonnegative $n\times n$ matrix over the integers. Let $r(A) =
r,\lambda_1,\ldots,\lambda_\ell$ be the distinct eigenvalues of $A$, and let $d$ be the size of the
dominating Jordan block of $A$. Then $A^m\in\Theta(r^mm^{d-1})$.
\end{lemma}
\begin{proof}
The theorem trivially holds for $r = 0$. Assume $r\geq 1$. Without loss of generality, we can
assume that $A$ does not have an eigenvalue $\lambda$ such that $\lambda \neq r$ and $|\lambda|=r$;
if such an eigenvalue exists, replace $A$ by $A^h$ (see Theorem~\ref{thm:P_F}). Let $J$ be the
Jordan canonical form of $A$, i.e., $A = SJS^{-1}$, where $S$ is a nonsingular matrix, and $J$ is a
diagonal block matrix of Jordan blocks. We use the following notation: $J_{\lambda, e}$ is a Jordan
block of order $e$ corresponding to eigenvalue $\lambda$, and $O_x$ is a square matrix, where all
entries are zero, except for $x$ at the top-right corner. Let $J_{r,d}$ be the dominating Jordan
block of $A$. It can be verified by induction that
$$
J_{r,d}^m = \left(
  \begin{array}{cccccc}
    r^m & {m\choose 1}r^{m-1} & {m\choose 2}r^{m-2} & \cdots & {m\choose d-2}r^{m-d+2} & {m\choose d-1}r^{m-d+1} \\
    0 & r^m & {m\choose 1}r^{m-1} & \cdots& {m\choose d-3}r^{m-d+3} & {m\choose d-2}r^{m-d+2} \\
    \vdots & \vdots & \vdots  & & \vdots & \vdots \\
    0 & 0 & 0 & \cdots & r^m & {m\choose 1}r^{m-1} \\
    0 & 0 & 0 & \cdots & 0& r^m \\
  \end{array}
\right).
$$
Thus the first row of $J_{r,d}^m$ has the form
$$
r^m\left[1 \;\;\; \frac{m}{r} \;\;\; \frac{m(m-1)}{2!r^2} \;\;\;\cdots\;\;\;
\frac{m(m-1)\cdots(m-(d-2))}{(d-1)!r^{d-1}}\right],
$$
and so
$$
\lim_{m\rightarrow\infty}\frac{J_{r,d}^m}{r^mm^{d-1}} = O_\alpha, \;\; \textrm{ where } \alpha =
\frac{1}{(d-1)!r^{d-1}}\;.
$$
All Jordan blocks other than the dominating block converge to zero blocks. and
$$
\lim_{m\rightarrow\infty}\frac{A^m}{r^mm^{d-1}} =
S\lim_{m\rightarrow\infty}\frac{J^m}{r^mm^{d-1}}S^{-1}.
$$
The result follows.
\end{proof}

\textbf{Note:} The growth order of $A^m$ supplies an algebraic proof of the fact that regular
languages can grow either polynomially or exponentially, but no intermediate growth order is
possible.  This result can also be derived from a more general matrix theoretic result of
Bell \cite{Bel05}.

Lemma~\ref{lem:growthorder} implies that to determine the order of growth of $L(M)$, we need to
compute the Perron-Frobenius eigenvalue $r$ of $A(M)$: if $r = 0$, then $L(M)$ is finite; if $r =
1$, the order of growth is polynomial; if $r> 1$, the order of growth is exponential. In the
polynomial case, if we want to determine the order of polynomial growth, we need to also compute
the size of the dominating Jordan block, which is the algebraic multiplicity of $r$ in the minimal
polynomial of $A(M)$.

Both computations can be done in polynomial time, though the runtime is more than cubic. The
characteristic polynomial, $c_A(x)$, can be computed in $\tilde O(n^4 \log \norm{A})$ bit
operations (here $\tilde O$ stands for soft-$O$, and $\norm{A}$ stands for the $L_\infty$ norm of
$A$). If $c_A(x) = x^n$ then $r = 0$; else, if $c_A(1)\neq 0$, then $r > 1$. In the case of $c_A(1)
= 0$, we need to check whether $c_A(x)$ has a real root in the open interval $(1,\infty)$. This can
be done using a real root isolation algorithm; it seems the best deterministic one uses $\tilde
O(n^6\log^2\norm{A})$ bit operations \cite{ESY06}. The minimal polynomial, $m_A(x)$, can be
computed through the rational canonical form of $A$ in $\tilde O(n^5 \log \norm{A})$ bit operations
(see references in \cite{GS02}). All algorithms mentioned above are deterministic; both $c_A(x)$
and $m_A(x)$ can be computed in $\tilde O(n^{2.697263} \log \norm{A})$ bit operations using a
randomized Monte Carlo algorithm \cite{KV04}.

An interesting problem is the following: given a nonnegative integer matrix $A$, is it possible to
decide whether $r(A) > 1$ in time better than $\tilde O(n^6\log^2\norm{A})$? Using our
combinatorial algorithm, we can do it in time $O(n^4\norm{A})$, by interpreting $A$ as the
adjacency matrix of a DFA over an alphabet of size $\norm{A}$, and applying the algorithm to each
of the connected components of $A$ separately. It would be interesting to find an algorithm
polynomial in $\log\norm{A}$.

%\section{Stuff left to do}
%
%Prove the following:
%
%\begin{proposition}
%Let $\alpha > 1$ be a real number and let $L \subseteq \Sigma^*$
%be a regular language.  The set
%$$\{n \in \mathbb{N} : |L \cap \Sigma^n| > \alpha^n \}$$
%is an ultimately periodic set.
%\end{proposition}
%
%\medskip
%
%Prove the analogue of Theorem~\ref{order} for NFA's, or show that the problem
%is intractable.

\section*{Acknowledgments}
We would like to thank Arne Storjohann for his input regarding algorithms for computing the
Perron-Frobenius eigenvalue of a nonnegative integer matrix.

\end{document}